\documentclass[a4paper,11pt,reqno]{amsart}
\usepackage[latin1]{inputenc}
\usepackage[T1]{fontenc}
\usepackage[english]{babel}
\usepackage{a4wide}
\usepackage{amsmath}
\usepackage{graphicx}
\usepackage{amssymb}
\usepackage[all]{xy}

\numberwithin{equation}{section}

\title{Irreversibility, least action principle and causality}
\author{Jacky Cresson $^{1,2}$, Pierre Inizan $^{2}$}
\address{$^1$ Laboratoire de Math\'ematiques Appliqu\'ees de Pau, Universit\'e de Pau et des Pays de l'Adour, avenue de l'Universit\'e, BP 1155, 64013 Pau Cedex, France}
\address{$^2$ Institut de M\'ecanique C\'eleste et de Calcul des \'Eph\'em\'erides, Observatoire de Paris, 77 avenue Denfert-Rochereau, 75014 Paris, France}

\begin{document}

\maketitle


\newcommand{\Dpar}[2]{\dfrac{\partial #1}{\partial #2}}
\newcommand{\dpar}[2]{\frac{\partial #1}{\partial #2}}
\newcommand{\Rk}[1]{\mathbb{R}^{#1}}
\newcommand{\R}[0]{\mathbb{R}}
\newcommand{\N}[0]{\mathbb{N}}
\newcommand{\Drl}[3]{\, {}_{#1} \mathcal{D}_{#2}^{#3} \,}
\newcommand{\Dl}[0]{\, \mathcal{D}^{\alpha} \,}
\newcommand{\Dr}[0]{\, \mathcal{D}^{\alpha}_* \,}
\newcommand{\Dli}[1]{\, \mathcal{D}^{#1} \,}
\newcommand{\Dri}[1]{\, \mathcal{D}^{#1}_* \,}
\newcommand{\mc}[1]{\mathcal{#1}}
\newcommand{\Dfp}[0]{d^+_\varepsilon}
\newcommand{\Dfm}[0]{d^-_\varepsilon}

\newtheorem{theorem}{Theorem}
\newtheorem{definition}{Definition}
\newtheorem{lemma}{Lemma}
\newtheorem{corollary}{Corollary}
\newtheorem{remark}{Remark}
\newtheorem{example}{Example}


\begin{abstract}
The least action principle, through its variational formulation, possesses a finalist aspect. It explicitly appears in the fractional calculus framework, where Euler-Lagrange equations obtained so far violate the causality principle. 
In order to clarify the relation between those two principles, we firstly remark that the derivatives used to described causal physical phenomena are in fact left ones. This leads to a formal approach of irreversible dynamics, where forward and backward temporal evolutions are decoupled.
This formalism is then integrated to the Lagrangian systems, through a particular embedding procedure.
In this set-up, the application of the least action principle leads to distinguishing trajectories and variations dynamical status.
More precisely, when trajectories and variations time arrows are opposed, we prove that the least action principle provides causal Euler-Lagrange equations, even in the fractional case.
Furthermore, the embedding developped is coherent. 
\end{abstract}

\bigskip



\section{Introduction}

The link between the least action principle and the causalilty principle has always been ambiguous. Poincar\'e \cite{Poincare} sumed it up as follows:
\medskip

\begin{quote}
\emph{
L'énoncé du principe de moindre action a quelque chose de choquant pour l'esprit. Pour se rendre d'un point à un autre, une molécule matérielle, soustraite à l'action donnée de toute force, mais assujettie à se mouvoir sur une surface, prendra la ligne géodésique, c'est-à-dire le chemin le plus court. Cette molécule semble connaître le point où on veut la mener, prévoir le temps qu'elle mettra à l'atteindre en suivant tel ou tel chemin, et choisir ensuite le chemin le plus convenable. L'énoncé nous la présente pour ainsi dire comme un être animé et libre. Il est clair qu'il vaudrait mieux le remplacer par un énoncé moins choquant, et où, comme diraient les philosophes, les causes finales ne sembleraient pas se substituer aux causes efficientes 
\footnote{
The very enunciation of the principle of least action is objectionable. To move from one point to another, a material molecule, acted upon by no force, but compelled to move on a surface, will take as its path the geodesic line - i.e., the shortest path. This molecule seems to know the point to which we want to take it, to foresee the time that it will take it to reach it by such a path, and then to know how to choose the most convenient path. The enunciation of the principle presents it to us, so to speak, as a living and free entity. It is clear that it would be better to replace it by a less objectionable enunciation, one in which, as philosophers would say, final effects do not seem to be substituted for acting causes \cite{Poincare_en}. 
}.
}
\end{quote}

\medskip

\begin{flushright}
Henri Poincaré, \emph{La science et l'hypothèse}, 1902.
\end{flushright}

\medskip
 
Therefore, how is it possible to obtain causal equations, i.e. equations taking into account only the past states, by using a principle which depends on the whole temporal interval? What does the information on the future become? Moreover, several approaches have been developped \cite{Agrawal, Baleanu_Agrawal, Cresson, Riewe_96} to generalize the least action principle and the Euler-Lagrange equation to the fractional case. In this formalism, the derivatives are non-local, which makes the past and future of the functions appear explicitly. Because of the simultaneous presence of left and right derivatives, none of those equations respect the causality principle. This difficulty may have been seen as a definite failure, and alternatives have notably been proposed in \cite{Dreisigmeyer, Stan} to get around this problem. However, because we believe that the least action principle should remain fundamental in any formalism, we choose to keep this approach - particularly the one in \cite{Cresson} -, and we prove in this paper that causality may be respected.
To this purpose, we formulate the following remark: when one observe a phenomena which one wants to describe using a differential equation, one only have access to the left derivatives of the functions, i.e. to the differential operators dependent on the past values of the function. Thus, this characteritic represents a trace of the time arrow, and the differential equation becomes attached to the forward temporal direction.
If we assume the existence of a similar differential equation, but related to the backward temporal evolution, we obtain a formal approach of irreversibility. Difficulties about causality inherent to the least action principle can be solved using this formalism. More precisely, we show that using a new embedding, termed \emph{asymmetric embedding}, it is possible to obtain causal Euler-Lagrange equations. In doing so, we observe that the information on the future lies in fact in the variations used by this variational method. The virtual status of these could hence moderate the finalist aspect of the least action principle. In addition, we prove that the asymmetric embedding is coherent, i.e. that this procedure is globally compatible with the least action principle.

The formal approach on irreversibility is first adressed in section \ref{Section:Irrev}. The asymmetric embbeding is introduced in section \ref{Section:Plgt_op} and applied to the Lagrangian systems in section \ref{Section:Plgt_L}. This leads to a causal Euler-Lagrange equation, obtained in section \ref{Section:Caus_Coh}. Application of this formalism to few examples of derivatives is reported in section \ref{Section:Cas_part} while those results are discussed in section \ref{Section:Concl}.


\section{An approach of irreversibility} \label{Section:Irrev}

\subsection{Dynamics and causality}

In physics, the causality principle means that the state of a system at a time $t$ is completely determined by its past, i.e. by its states at times $t'$, $t' < t$. Therefore, if a system is described by variables $x \in \R^n$, the variations of those ones (for example the velocities) should only depend on the past instants. The derivatives, which express those variations, are hence \emph{left} derivatives. They will be denoted by $\mc{D}^+$. The following definition formalises these ideas.

\begin{definition}
The evolution of a system is said causal in the direction \emph{past $\rightarrow$ future} if it can be written as
\begin{equation} \label{eq_futur}
f_+ \left( x(t), \mc{D}^+ x(t), \ldots, (\mc{D}^+)^k x(t), t \right) = 0.
\end{equation} 
\end{definition}

A first example of operator $\mc{D}^+$ is the usual left derivative
\begin{equation} \label{dp}
d_+ x(t) = \lim_{
\begin{array}{c}
\scriptstyle \varepsilon \rightarrow 0 \\
\scriptstyle \varepsilon > 0
\end{array}}
\frac{x(t) - x(t - \varepsilon)}{\varepsilon}.
\end{equation} 

However, following \cite{Hilfer_RFT}, we postulate that the evolution of a general physical system is \emph{a priori} irreversible: the dynamics in the direction \emph{future $\rightarrow$ past} cannot be described by \eqref{eq_futur}. In order to have a complete description of the system, a supplementary differential equation has to be introduced which accounts for the evolution towards the past. The evolution operators are in this case right derivatives, denoted by $\mc{D}^-$, and lead to the following definition.

\begin{definition}
The evolution of a system is said causal in the direction \emph{future $\rightarrow$ past} if it can be written as
\begin{equation} \label{eq_passe}
f_- \left( x(t),\mc{D}^- x(t), \ldots, (\mc{D}^-)^l x(t), t \right) = 0.
\end{equation}
\end{definition}

For the previous example, the operator $\mc{D}^-$ is the usual right derivative
\begin{equation*}
d_- x(t) = \lim_{
\begin{array}{c}
\scriptstyle \varepsilon \rightarrow 0 \\
\scriptstyle \varepsilon > 0
\end{array}}
\frac{x(t + \varepsilon) - x(t)}{\varepsilon}.
\end{equation*}

Now we can precise our formal approach of irreversibility.
\begin{definition}
A system is said reversible if \eqref{eq_futur} and \eqref{eq_passe} have the same solutions. Otherwise, it is said irreversible.
\end{definition}

\begin{remark}
We emphasize the formal aspect of this definition: our goal here is not to understand the physical origin of irreversibility. This problem, related to Boltzmann's work on entropy in the 1870 decade, is still discussed in the physics community. We simply mention that recent answers have been proposed through chaotic systems \cite{Prigogine_99, Zasl_PT}. 
\end{remark}

Even if our approach is formal, it will be useful to understand better the least action principle, more precisely its relation with causality. 

\subsection{Asymmetric dynamical representation}

In the rest of the paper, we will consider a system evolving in $\R^n$, during a temporal interval $[a,b]$. This system will de denoted by $\mc{S}$.

We introduce the vector space $\mc{U}$ defined by
\begin{equation*}
\mc{U} = \{ x \in C^0([a,b],\R^n) \: | \: \mc{D}^\pm x \in C^0([a,b],\R^n) \}.
\end{equation*} 
In the rest of the paper, we will only consider trajectories which belong to $\mc{U}$.
Moreover, we suppose that $C^\infty_c([a,b],\R^n) \subset \mc{U}$, where $C^\infty_c([a,b],\R^n)$ is the set of $C^\infty$ functions with compact support in $[a,b]$.

First we precise the notion of dynamics in this approach.

\begin{definition} \label{definition:asy_dyn_rep}
The asymmetric dynamical representation of a system $\mc{S}$ is defined by the couple $(x_+, x_-) \in \R^n \times \R^n$ and by their respective temporal evolutions governed by the following differential equations
\begin{eqnarray*}
f_+(x_+(t), \mc{D}^+ x_+(t), \ldots, (\mc{D}^+)^k x_+(t), t) & = & 0, \\
f_-(x_-(t), \mc{D}^- x_-(t), \ldots, (\mc{D}^-)^l x_-(t), t) & = & 0.
\end{eqnarray*} 
The variable $x_+$ represents the evolution in the direction \emph{past $\rightarrow$ future}, and $x_-$ in the direction \emph{future $\rightarrow$ past}.
\end{definition}

The time arrow is hence characterised by two objects: the global structure of the differential equation, via $f_\pm$, and the temporal evolution operator, i.e. $\mc{D}^\pm$. Whereas the direction of evolution clearly appears in the second one, it is not the case for the other.  That is why we will now suppose that $f_+ = f_- = f$. In this case, the direction of the time arrow becomes exclusively determined by the choice of the derivative, and the equations of the dynamics are
\begin{eqnarray}
f(x_+(t), \mc{D}^+ x_+(t), \ldots, (\mc{D}^+)^k x_+(t), t) & = & 0 \label{fpb}, \\
f(x_-(t), \mc{D}^- x_-(t), \ldots, (\mc{D}^-)^k x_-(t), t) & = & 0 \label{fmb}.
\end{eqnarray} 
 
\begin{remark}
When $\mc{D}^\pm = d_\pm$, the trajectory is often considered differentiable, i.e. it verifies $d_+ x(t) = d_- x(t) = d /dt \: x(t)$. All information on the time arrow is therefore lost. 
\end{remark}

In order to deal with the two directions of evolution in a unified way, we introduce the following functional spaces
\begin{eqnarray*}
\mc{V} & = & \mc{U} \times \mc{U}, \\
\mc{V}^+ & = & \mc{U} \times \{ 0 \}, \\
\mc{V}^- & = & \{ 0 \} \times \mc{U}. \\
\end{eqnarray*} 

For $X = (x_+,x_-) \in \mc{V}$, we define the differential operator $\mc{D}$ by
\begin{equation*}
\mc{D} \, X = (\mc{D}^+ x_+, \mc{D}^- x_-).
\end{equation*} 

Consequently, for $k \in \N^*$, $\mc{D}^k \, X = \left( (\mc{D}^+)^k x_+, (\mc{D}^-)^k x_- \right)$.

This approach will now be integrated in the framework of the embedding theories (see \cite{Cresson, Cresson_Darses}).


\section{Asymmetric embedding of differential operators} \label{Section:Plgt_op}

The initial motivation for the framework presented above is to conciliate fractional least action principle and causality. In \cite{Agrawal, Baleanu_Agrawal, Cresson, Riewe_96}, the Euler-Lagrange equations stemming from a least action principle contain both the operators $\mc{D}^+$ and $\mc{D}^-$. They violate the causality principle, as it has been noticed in \cite{Dreisigmeyer}. For example, in \cite{Cresson}, 
the equations obtained lead to 
\begin{equation}
\partial_1 L(x(t), \mc{D}^+ x(t), t) - \mc{D}^- \partial_2 L(x(t), \mc{D}^+ x(t), t) = 0, \nonumber
\end{equation}
or to
\begin{equation}
\partial_1 L(x(t), \mc{D}^- x(t), t) - \mc{D}^+ \partial_2 L(x(t), \mc{D}^- x(t), t) = 0. \nonumber
\end{equation} 

In the framework of embedding theories, this simultaneous presence of the two operators is problematic for the coherence \cite{Cresson}: the embedding procedure and the least action principle are not cummutative. A solution proposed in \cite{Cresson} and exploited in \cite{Inizan_CBHF} consisted in restricting the space of variations used in the least action principle. However, despite being very strong, those constraints do not however lead to a unique solution.

A new embedding procedure, the \emph{asymmetric embedding}, is presented here, and solves those problems. Its validity is not restricted to the fractional case.

\bigskip

We begin with the differential operators introduced in \cite{Cresson}. 

For two vector spaces $A$ and $B$, we denote $\mc{F}(A,B)$ the vector space of the functions $f \, : \, A \rightarrow B$.
If $f \in \mc{F}(\R^{n(k+1)} \times \R, \R^m)$, we define an associated operator 
\begin{equation} \label{operateur}
F \, : \, y \in \mc{F}([a,b], \R^{n(k+1)}) \longmapsto f(y(\bullet),\bullet),
\end{equation}  
where $f(y(\bullet),\bullet)$ is defined by
\begin{equation}
f(y(\bullet),\bullet) \, : \, t \in [a,b] \longmapsto f(y(t),t). \nonumber
\end{equation} 

If $\textbf{f} = \{ f_i \}_{0 \leq i \leq p}$ and $\textbf{g} = \{ g_j \}_{1 \leq j \leq p}$ are two families of $\mc{F}(\R^{n(k+1)} \times \R, \R^m)$, we introduce the operator $\mc{O}_\textbf{f}^\textbf{g}$ defined by 
\begin{equation} \label{Ofg}
\mc{O}_\textbf{f}^\textbf{g} \, : \, x \in \mc{U} \longmapsto \left[ F_0 + \sum_{i=1}^p F_i \cdot \frac{d^i}{dt^i} \circ G_i \right] \left( x(\bullet), \ldots,  \frac{d^k}{dt^k} x(\bullet), \bullet \right),
\end{equation} 
where, for two operators $A = (A_1, \ldots, A_m)$ and $B = (B_1, \ldots, B_m)$, $A \cdot B$ is defined by
\begin{equation*}
(A \cdot B)(y) = \left( A_1(y) B_1(y), \ldots, A_m(y) B_m(y) \right).
\end{equation*} 

We extend now those operators to deal with both evolution directions.

\begin{definition}
With the previous notations, the asymmetric representation of operator $F$, denoted by $\tilde{F}$, is defined by
\begin{equation} \label{tilde}
\tilde{F} \, : \, (y_1, y_2) \in {\mc{F}\left([a,b], \R^{n(k+1)}\right)}^2 \longmapsto f(y_1(\bullet) + y_2(\bullet),\bullet).
\end{equation} 
\end{definition}

Let $\mc{M}_{m,2m}(\R)$ be the set of real matrices with $m$ rows and $2m$ columns. We note $I_m$ the identity matrix of dimension $m$, and we introduce the operator $\sigma$ defined by
\begin{equation}
\begin{array}{cccl}
\sigma \, : & \mc{V} & \longrightarrow & \quad \mc{M}_{m,2m}(\R) \\
               & X & \longmapsto & (I_m \quad 0) \quad \text{if } X \in \mc{V}^+ \backslash \{0 \}, \\ 
		&	&        & (0 \quad I_m) \quad \text{if } X \in \mc{V}^- \backslash \{0 \}, \\ 
		&	&        & (I_m \; I_m) \quad \text{otherwise}.
\end{array}
\end{equation} 

Now we can define the asymmetric embedding of an operator.

\begin{definition} \label{definition:plgt}
With the previous notations, the asymmetric embedding of operator \eqref{Ofg}, denoted by $\mc{E}(\mc{O}_\textbf{f}^\textbf{g})$, is defined by
\begin{equation} \label{plgt}
\mc{E}(\mc{O}_\textbf{f}^\textbf{g}) \, : \, X \in \mc{V} \longmapsto \left[ \tilde{F}_0 + \sigma(X) \sum_{i=1}^p 
\begin{pmatrix}
\tilde{F}_i \cdot (\mc{D}^+)^i \circ \tilde{G}_i \\
\tilde{F}_i \cdot (\mc{D}^-)^i \circ \tilde{G}_i
\end{pmatrix}
\right] \left( X(\bullet), \ldots,  \mc{D}^k X(\bullet), \bullet \right).
\end{equation} 
\end{definition}

In particular, for $(x_+,0) \in \mc{V}^+$, \eqref{plgt} becomes
\begin{equation*} 
\mc{E}(\mc{O}_\textbf{f}^\textbf{g})(x_+,0)(t) = \left[ F_0 + \sum_{i=1}^p F_i \cdot (\mc{D}^+)^i \circ G_i \right] \left( x_+(t), \ldots,  (\mc{D}^+)^k x_+(t), t \right),
\end{equation*}

and for $(0,x_-) \in \mc{V}^-$, we have
\begin{equation*}
\mc{E}(\mc{O}_\textbf{f}^\textbf{g})(0,x_-)(t) = \left[ F_0 + \sum_{i=1}^p F_i \cdot (\mc{D}^-)^i \circ G_i \right] \left( x_-(t), \ldots,  (\mc{D}^-)^k x_-(t), t \right).
\end{equation*}

\begin{example}

We set $n=m=p=1$, $k=2$.
Let $f_0, \, f_1, \, g_1 \, : \, \R^3 \times \R \longrightarrow \R$ be three functions defined by
\begin{eqnarray*}
f_0(a,b,c,t) & = & c + e^{-t} \cos b,  \\
f_1(a,b,c,t) & = & 1,  \\
g_1(a,b,c,t) & = & \cos a. 
\end{eqnarray*}

The associated operator $\mc{O}_\textbf{f}^\textbf{g}$ verifies 
\begin{equation*}
\mc{O}_\textbf{f}^\textbf{g}(x)(t) = \frac{d^2}{dt^2} x(t) + e^{-t} \cos \left( \frac{d}{dt} x(t) \right) + \frac{d}{dt} \cos(x(t)).
\end{equation*}

Its asymmetric embedding $\mc{E}(\mc{O}_\textbf{f}^\textbf{g})$ is defined for trajectories $(x_+,x_-) \in \mc{V}$, and is given by
\begin{align*}
\mc{E}(\mc{O}_\textbf{f}^\textbf{g})(x_+,x_-)(t) = (\mc{D}^+)^2 x_+(t) + (\mc{D}^-)^2 x_-(t) + e^{-t} \cos(\mc{D}^+ x_+(t) + \mc{D}^- x_-(t)) \\
+ \sigma(x_+,x_-) 
\begin{pmatrix}
\mc{D}^+ \cos(x_+(t) + x_-(t))  \\
\mc{D}^- \cos(x_+(t) + x_-(t))
\end{pmatrix}.
\end{align*} 

For $(x_+,0) \in \mc{V}^+$, the embedding becomes
\begin{equation*}
\mc{E}(\mc{O}_\textbf{f}^\textbf{g})(x_+,0)(t) = (\mc{D}^+)^2 x_+(t) + e^{-t} \cos(\mc{D}^+ x_+(t)) + \mc{D}^+ \cos(x_+(t)), 
\end{equation*}
and for $(0,x_-) \in \mc{V}^-$, we have
\begin{equation*}
\mc{E}(\mc{O}_\textbf{f}^\textbf{g})(0,x_-)(t) = (\mc{D}^-)^2 x_-(t) + e^{-t} \cos(\mc{D}^- x_-(t)) + \mc{D}^- \cos(x_-(t)).
\end{equation*}

\end{example}

The ordinary differential equations may be written by using operators $\mc{O}_\textbf{f}^\textbf{g}$. Following \cite{Cresson}, we consider the differential equations of the form 
\begin{equation} \label{eq_diff}
\mc{O}_\textbf{f}^\textbf{g}(x) = 0, \quad x \in \mc{U}.
\end{equation} 

\begin{definition} \label{definition:eq_plgt}
With the previous notations, the asymmetric embedding of differential equation \eqref{eq_diff} is defined by 
\begin{equation} \label{eq_plgt}
\mc{E}(\mc{O}_\textbf{f}^\textbf{g})(X) = 0, \quad X \in \mc{V}.
\end{equation} 
\end{definition}

Consequently, if $(x_+,0) \in \mc{V}^+$, \eqref{eq_plgt} becomes
\begin{equation*}
\left[ F_0 + \sum_{i=1}^p F_i \cdot (\mc{D}^+)^i \circ G_i \right] \left( x_+(t), \ldots,  (\mc{D}^+)^k x_+(t), t \right) = 0,
\end{equation*}
and plays the part of \eqref{fpb}.

Similarly, for $(0,x_-) \in \mc{V}^-$, we obtain 
\begin{equation*}
\left[ F_0 + \sum_{i=1}^p F_i \cdot (\mc{D}^-)^i \circ G_i \right] \left( x_-(t), \ldots,  (\mc{D}^-)^k x_-(t), t \right) = 0,
\end{equation*}
which may be related to \eqref{fmb}.

This method is now applied to the Lagrangian systems.


\section{Asymmetric embedding of Lagrangian systems} \label{Section:Plgt_L}

We consider the same system $\mc{S}$ as above, but we suppose now that it admits a differentiable Lagrangian $L \in \mc{F}(\R^{2n} \times \R, \R)$. 
In the rest of this paper, the Lagrangian $L$ and its associated operator \eqref{operateur} will be identified.
For such systems, the least action principle stipulates that the extrema of the action provide the equation of the dynamics, called Euler-Lagrange equation. 

As we show later on, the causality problem lies in the integration by parts which appears in the calculus of variations.
From now on, we suppose that the operators $\mc{D}^+$ et $\mc{D}^-$ verify
\begin{equation} \label{IPP}
\int_a^b \mc{D}^+ f(t) \, g(t) \, dt = - \int_a^b f(t) \, \mc{D}^- g(t) \, dt + R_{ab}(f,g),
\end{equation} 
where $f, \, g \in \mc{U}$, and where $R_{ab}(f,g)$ contains the evaluations of $f$ and $g$ (and possibly their derivatives), at points $a$ and $b$. 

For example, if $\mc{D}^+ = \mc{D}^- = d/dt$, $R_{ab}(f,g) = f(b) g(b) - f(a) g(a)$.

Concerning the reference dynamics linked to the operator $d/dt$, the action associated to $L$, denoted by $\mc{A}(L)$, is defined by 
\begin{equation} \label{action}
\begin{array}{cccl}
\mc{A}(L) \, : & C^1([a,b],\R^n) & \longrightarrow & \quad \R  \\
               &   x    & \longmapsto & \displaystyle{\int_a^b} L \left( x(t), \frac{d}{dt} x(t), t \right) \, dt,
\end{array}
\end{equation}

Its extrema provide the Euler-Lagrange equation \cite{Arnold}
\begin{equation} \label{EL}
\partial_1 L(x(t), \frac{d}{dt} x(t), t) - \frac{d}{dt} \partial_2 L(x(t), \frac{d}{dt} x(t), t) = 0.
\end{equation}

\begin{remark}
For Lagrangian systems, the only Lagrangian determines the global structure of the equations. The dynamics is then completely fixed by the choice of the temporal evolution operator. Therefore the assumption $f_+ = f_-$ in \eqref{fpb} and \eqref{fmb} seems justified for those systems.
\end{remark}

This formalism will now be fit to the approach on irreversibility by using the asymmetric embedding.

\subsection{Asymmetric embedding of the Euler-Lagrange equation}

For $X=(x_1,x_2), Y=(y_1,y_2) \in \R^{2n}, \, t \in \R$, the asymmetric representation of the operator $L$, denoted by $\tilde{L}$, verifies 
\begin{equation*}
\tilde{L}(X, Y, t) = L(x_1 + x_2, y_1 + y_2, t).
\end{equation*} 

A first idea consists in embedding \eqref{EL} directly. Given that
\begin{equation*}
\Dpar{L}{x_1}(x_1 + x_2, y_1 + y_2, t) = \Dpar{L}{x_2}(x_1 + x_2, y_1 + y_2, t) = \partial_1 L(x_1 + x_2, y_1 + y_2, t),
\end{equation*} 

we note $ \partial_1 \tilde{L}(X,Y,t) = \partial_1 L(x_1 + x_2, y_1 + y_2, t)$. Similarly, we note $\partial_2 \tilde{L}(X,Y,t) = \partial_2 L(x_1 + x_2, y_1 + y_2, t)$. 

\begin{theorem}
For $X \in \mc{V}$, the asymmetric embedding of \eqref{EL} is defined by
\begin{equation} \label{EL_plgt}
\partial_1 \tilde{L}(X(t), \mc{D} X(t), t) + \sigma(X) 
\begin{pmatrix}
\mc{D}^+ \partial_2 \tilde{L}(X(t), \mc{D} X(t), t)  \\
\mc{D}^- \partial_2 \tilde{L}(X(t), \mc{D} X(t), t)
\end{pmatrix}
=
0
\end{equation} 

In particular, for $(x_+,0) \in \mc{V}^+$, \eqref{EL_plgt} becomes
\begin{equation} \label{EL_plgt_p}
\partial_1 L(x_+(t), \mc{D}^+ x_+(t), t) - \mc{D}^+ \partial_2 L(x_+(t), \mc{D}^+ x_+(t), t) = 0,
\end{equation} 

and for $(0,x_-) \in \mc{V}^-$, 
\begin{equation} \label{EL_plgt_m}
\partial_1 L(x_-(t), \mc{D}^- x_-(t), t) - \mc{D}^- \partial_2 L(x_-(t), \mc{D}^- x_-(t), t) = 0.
\end{equation} 

\end{theorem}

\begin{proof}
Equation \eqref{EL} may be written like \eqref{eq_diff} with $k=1$, $p=1$, $\textbf{f} = \{ \partial_1 L, 1\}$ and $\textbf{g} = \{ \partial_2 L\}$. We conclude by using definitions \ref{definition:plgt} and \ref{definition:eq_plgt}.
\end{proof}

Those two equations may play the part of \eqref{fpb} and \eqref{fmb} for the system $\mc{S}$. A possible asymmetric dynamical representation of $\mc{S}$ can hence be given by \eqref{EL_plgt_p}-\eqref{EL_plgt_m}.

\subsection{Calculus of the asymmetric variations}

However, it is also possible to embed the Lagrangian $L$ itself instead of \eqref{EL}. The least action may then be applied to the embedded action. 
Indeed, the Lagrangian $L$ may be written as \eqref{Ofg}, with $p=0$ and $\textbf{f}= \{ \tilde{L} \}$. Its asymmetric embedding, which will be noted $\hat{L}$, is defined by
\begin{equation}
\hat{L} \, : \, X \in \mc{V} \longmapsto \tilde{L}(X(\bullet), \mc{D} X(\bullet), \bullet).
\end{equation} 

The associated action \eqref{action} is now given by
\begin{equation} \label{action_plgt}
\begin{array}{cccl}
\mc{A}(\hat{L}) \, : & \mc{V} & \longrightarrow & \quad \R  \\
               &   X    & \longmapsto & \displaystyle{\int_a^b} \tilde{L} \left( X(t), \mc{D} X(t), t \right) \, dt.
\end{array}
\end{equation}

Similarly to the classical least action principle \cite{Arnold}, the equation of motion will be characterized through the extremal of \eqref{action_plgt}.
We need here to precise the notion of extremum.

\medskip

Let $A$ be a vector space, $B$ a subspace of $A$, and $f \, : \, A \rightarrow \R$ a functional. Let $x \in A$. 

\begin{definition}
The functional $f$ has a $B$-minimum (respectively $B$-maximum) point at $x$ if for all $h \in B$, $f(x+h) \geq f(x)$ (respectively $f(x+h) \leq f(x)$).
The functional $f$ has a $B$-extremum point at $x$ if it has a $B$-minimum point or a $B$-maximum point at $x$.
\end{definition}

In the differentiable case, the classical necessary condition remains with this definition.

\begin{lemma} \label{lemma:CN}
We suppose that $f$ is differentiable. If $f$ has a $B$-extremum point at $x \in A$, then for all $h \in B$, $df(x)(h)=0$, where $df(x)$ is the differential of $f$ at $x$.
In this case, $x$ is called a $B$-extremal of $f$.
\end{lemma}

\begin{proof}
Let $h \in B$. The differentiable function $f_h \, : \, t \mapsto f(x+th)$ has an extremum point in $0$. Therefore $f_h'(0) = df(x)(h) = 0$.
\end{proof}

Now we introduce the space of variations
\begin{equation}
\mc{H} = \left\{ (h_+,h_-) \in \mc{V} \; | \; \forall f \in \mc{U}, \, R_{ab}(h_+,f) = R_{ab}(f,h_-) = 0 \right\}.
\end{equation} 

Without further assumptions, we obtain the following result, derived from \cite[theorem 3.11]{Avez}:

\begin{theorem} \label{theorem:EL_gnl}
Let $X \in \mc{V}$. We suppose that $t \mapsto \partial_2 \tilde{L} (X(t), \mc{D} X(t), t) \in \mc{U}$. Then we have the following equivalence:

the function $X \in \mc{V}$ is a $\mc{H}$-extremal of the action $\mc{A}(\hat{L})$ if and only if it verifies 
\begin{eqnarray} 
\partial_1 \tilde{L}(X(t), \mc{D} X(t), t) - \mc{D}^- \partial_2 \tilde{L} (X(t), \mc{D} X(t), t) & = & 0, \label{EL_gnle1} \\
\partial_1 \tilde{L}(X(t), \mc{D} X(t), t) - \mc{D}^+ \partial_2 \tilde{L} (X(t), \mc{D} X(t), t) & = & 0, \label{EL_gnle2}
\end{eqnarray}
for all $t \in [a,b]$.

\end{theorem}

\begin{proof}
The Lagrangian $L$ being differentiable, $\mc{A}(\hat{L})$ is also differentiable. From lemma \ref{lemma:CN}, $X=(x_+,x_-) \in \mc{V}$ is a $\mc{H}$-extremal of $\mc{A}(\hat{L})$ if and only if for all $H \in \mc{H}$, $d\mc{A}(\hat{L})(X)(H) = 0$. For $H=(h_+,h_-) \in \mc{H}$, we have
\begin{equation*}
\begin{array}{l}
\mc{A}(\hat{L})(X+H)  = \displaystyle{\int_a^b} \tilde{L}( (X+H)(t), \mc{D} (X+H)(t), t) \, dt, \\
	\qquad = \displaystyle{\int_a^b} L((x_+ + x_- + h_+ + h_-)(t), \mc{D}^+ (x_+ + h_+)(t) + \mc{D}^- (x_- + h_-)(t),t) \, dt, \\
	\qquad   = \mc{A}(\hat{L})(X) + \displaystyle{\int_a^b} \partial_1 \tilde{L}(X(t), \mc{D} X(t),t) \, (h_+(t) + h_-(t)) \, dt \\
	\qquad   \qquad + \displaystyle{\int_a^b} \partial_2 \tilde{L}(X(t), \mc{D} X(t),t) \, (\mc{D}^+ h_+(t) + \mc{D}^- h_-(t)) \, dt + o(H), \\
	\qquad  =  \mc{A}(\hat{L})(X) + \displaystyle{\int_a^b} \left[ \partial_1 \tilde{L}(X(t), \mc{D} X(t),t) \, h_+(t) + \partial_2 \tilde{L}(X(t), \mc{D} X(t),t) \, \mc{D}^+ h_+(t) \right] dt \\
	\qquad  \qquad + \displaystyle{\int_a^b} \left[ \partial_1 \tilde{L}(X(t), \mc{D} X(t),t) \, h_-(t) + \partial_2 \tilde{L}(X(t), \mc{D} X(t),t) \, \mc{D}^- h_-(t) \right] dt + o(H). \\
\end{array}
\end{equation*} 

Consequently,
\begin{multline*}
d\mc{A}(\hat{L})(X)(H) = \int_a^b \left[ \partial_1 \tilde{L}(X(t), \mc{D} X(t),t) \, h_+(t) + \partial_2 \tilde{L}(X(t), \mc{D} X(t),t) \, \mc{D}^+ h_+(t) \right] dt \\
	+ \int_a^b \left[ \partial_1 \tilde{L}(X(t), \mc{D} X(t),t) \, h_-(t) + \partial_2 \tilde{L}(X(t), \mc{D} X(t),t) \, \mc{D}^- h_-(t) \right] dt.
\end{multline*} 

Given that $t \mapsto \partial_2 \tilde{L} (X(t), \mc{D} X(t), t) \in \mc{U}$, we can use \eqref{IPP}. No supplementary term appears because $H \in \mc{H}$.
Therefore the differential is given by
\begin{multline*}
d\mc{A}(\hat{L})(X)(H) = \int_a^b \left[ \partial_1 \tilde{L}(X(t), \mc{D} X(t),t)- \mc{D}^- \partial_2 \tilde{L}(X(t), \mc{D} X(t),t) \,  \right] h_+(t) \, dt \\
	+ \int_a^b \left[ \partial_1 \tilde{L}(X(t), \mc{D} X(t),t) - \mc{D}^+ \partial_2 \tilde{L}(X(t), \mc{D} X(t),t) \, \mc{D}^- \right] h_-(t) \, dt.
\end{multline*} 

Given that $X \in \mc{U}$, we have $t \mapsto \partial_1 \tilde{L}(X(t), \mc{D} X(t),t) \in C^0([a,b],\R^n)$. 
Consequently, $t \mapsto \partial_1 \tilde{L}(X(t), \mc{D} X(t),t) - \mc{D}^\pm \partial_2 \tilde{L}(X(t), \mc{D} X(t),t) \in C^0([a,b],\R^n)$.
The space $\mc{U}$ contains $C^\infty_c([a,b],\R^n)$, so we can apply \cite[theorem 1.2.4]{Hormander}: $d\mc{A}(\hat{L})(X)(H) = 0$ for all $H \in \mc{H}$ if and only if \eqref{EL_gnle1} and \eqref{EL_gnle2} are satisfied.
\end{proof}

When we look at the evolution towards future, i.e. at the evolution of $x_+$, we see that $(x_+,0)$ is a $\mc{H}$-extremal of the action if and only if it verifies
\begin{eqnarray} 
\partial_1 L(x_+(t), \mc{D}^+ x_+(t),t) - \mc{D}^- \partial_2 L (x_+(t), \mc{D}^+ x_+(t),t) & = & 0, \label{EL_gnle_pm} \\
\partial_1 L(x_+(t), \mc{D}^+ x_+(t),t) - \mc{D}^+ \partial_2 L (x_+(t), \mc{D}^+ x_+(t),t) & = & 0. \nonumber
\end{eqnarray} 

Equation \eqref{EL_gnle_pm} does not respect causality because of the operator $\mc{D}^-$. The trajectory $x_+$ should hence verify two different equations, whereas only the second seems acceptable from a physical point of view. The same problem arises for the evolution towards past.
It will now be shown that it is possible to overcome those difficulties by restricting the variations.


\section{Causality and coherence} \label{Section:Caus_Coh}

From a physical point of view, the elements of $\mc{V}$ are meaningless. Only their restrictions to $\mc{V}^+$ and $\mc{V}^-$ are relevant. The same remark applies to $\mc{H}$. Consequently, the space of variations has to be questioned. In \cite{Inizan_CBHF}, we propose to restrict the variations $h$ by assuming $\mc{D}^+ h = \mc{D}^- h$. Unfortunately, this hypothesis seems very strong, and may not be related to the dynamics. Moreover, a supplementary term appears in the Euler-Lagrange equation. In this paper we also restrict the variations but such problems will not arise.

If we study the evolution towards future (we would proceed likewise for the other direction), it would seem natural to only consider the variations which belong also to $\mc{V}^+$. Therefore we introduce a new space of variations $\mc{H}^+ = \mc{H} \cap \mc{V}^+$. Similarly, we set $\mc{H}^- = \mc{H} \cap \mc{V}^-$. From theorem \ref{theorem:EL_gnl}, we deduce the following result.

\begin{corollary}
Let $(x_+,0) \in \mc{V}^+$. We suppose that $t \mapsto \partial_2 L (x_+(t), \mc{D}^+ x_+(t), t) \in \mc{U}$. Then we have the following equivalence:

the function $(x_+,0)$ is a $\mc{H}^+$-extremal of the action $\mc{A}(\hat{L})$ if and only if $x_+$ verifies 
\begin{equation} \label{EL_p_anti}
\partial_1 L(x_+(t), \mc{D}^+ x_+(t), t) - \mc{D}^- \partial_2 L(x_+(t), \mc{D}^+ x_+(t), t) = 0,
\end{equation}
for all $t \in [a,b]$.
\end{corollary}

We obtain a single equation, which is precisely \eqref{EL_gnle_pm}, the problematic one. Consequently, the variations (in $\mc{H}$) and the trajectories (in $\mc{V}$) cannot have the same status.

On the contrary, if the trajectories are chosen in $\mc{H}^-$, the problem is solved.

\begin{theorem}
Let $(x_+,0) \in \mc{V}^+$. We suppose that $t \mapsto \partial_2 L (x_+(t), \mc{D}^+ x_+(t), t) \in \mc{U}$. Then we have the following equivalence:

the function $(x_+,0)$ is a $\mc{H}^-$-extremal of the action $\mc{A}(\hat{L})$ if and only if $x_+$ verifies 
\begin{equation} \label{EL_p}
\partial_1 L(x_+(t), \mc{D}^+ x_+(t), t) - \mc{D}^+ \partial_2 L(x_+(t), \mc{D}^+ x_+(t), t) = 0,
\end{equation}
for all $t \in [a,b]$.
\end{theorem}

We obtain once again a unique equation, but which is now causal in the direction \emph{past $\rightarrow$ future}. We have a similar result for the other direction.

\begin{corollary}
Let $(0,x_-) \in \mc{V}^-$. We suppose that $t \mapsto \partial_2 L (x_-(t), \mc{D}^- x_-(t), t) \in \mc{U}$. Then we have the following equivalence:

the function $(0,x_-)$ is a $\mc{H}^+$-extremal of the action $\mc{A}(\hat{L})$ if and only if $x_-$ verifies 
\begin{equation} \label{EL_m}
\partial_1 L(x_-(t), \mc{D}^- x_-(t), t) - \mc{D}^- \partial_2 L(x_-(t), \mc{D}^- x_-(t), t) = 0,
\end{equation}
for all $t \in [a,b]$.
\end{corollary}

This equation is causal in the direction \emph{future $\rightarrow$ past}.

Both of the equations \eqref{EL_p} and \eqref{EL_m} stem from a least action principle, and moreover respect the causality principle.
In addition, no significant restriction is done on the non-zero components of the variations. Therefore we obtain a least action principle similar to the classical one (with $d/dt$), except that trajectories and variations are not ruled by the same dynamics.

This result may seem surprising: it shows that the equations of the dynamics in a given temporal direction are obtained through variations evolving in the opposite way. Let us discuss this paradox.

The least action principle is a global vision of the dynamics: the trajectory is directly determined on its whole temporal interval $[a,b]$. According to Poincaré, the system ``seems to know the point to which we want to take it''. In \cite{Martin-Robine}, an history of this principle is presented, and it is shown that this finalist aspect has been the subject of controversies since its formulation by Maupertuis in 1746 \cite{Maupertuis_1746}. 
By making the past and the future of the trajectory appear explicitly, equation \eqref{EL_p_anti} is in agreement with this global approach.
For example, in the fractional framwork, $\mc{D}^+ = \Drl{a}{t}{\alpha}$ and takes into account the whole past ($[a,t]$) of the trajectory and $\mc{D}^- = - \Drl{t}{b}{\alpha}$, all the future ($[t,b]$).
So it may seem that variational formulation (global) and causal equation could be incompatible. But as it has just been shown, the variations can lift this difficulty. Because they obey to the reverse dynamics, they ``catch'' the anti-causal part of the least action principle. Then the equations of the dynamics can respect causality. In the sum $x_+ + h_-$, we add two comparable functions, but their underlying dynamical and physical natures differ. The trajectory $x_+$ may be called ``real'', \emph{actual}, while the variation $h_-$ may be seen as ``virtual'', \emph{potential}. To sum up, the finalist aspect of the least action principle may lie in the nature of the variations. Given that these do not possess a concrete realisation, this problematic characteristic of the principle may seem less disconcerting. 

\begin{remark}
In the case of a dynamics governed by $d/dt$, this discussion is obscured because of the local aspect of the derivative.
\end{remark}

Furthermore, we note that \eqref{EL_p} and \eqref{EL_m} are identical to \eqref{EL_plgt_p} and \eqref{EL_plgt_m}. So we obtain the following commutative diagram
\begin{equation*}
\xymatrix{
L(x(t),\frac{d}{dt} x(t),t) \ar[d]_{LAP} \ar[r]^{AE} & \tilde{L}(X(t),\mc{D} X(t),t) \ar[d]^{LAP} \\
\left[ \partial_1 L - \frac{d}{dt} \partial_2 L \right] (x(t),\frac{d}{dt} x(t),t) = 0  \ar[r]_{AE} & 
\left[ \partial_1 L - \mc{D}^\pm \partial_2 L \right](x_\pm(t), \mc{D}^\pm x_\pm(t), t) = 0,
}
\end{equation*}

where $LAP$ means ``least action principle'' and $AE$ ``asymmetric embedding''. Following \cite{Cresson}, we say that the asymmetric embedding is \emph{coherent}. 

This ``robustness'' motivates the adaptation of definition \ref{definition:asy_dyn_rep} for Lagrangian systems.

\begin{definition}
The asymmetric dynamical representation of a Lagrangian system with Lagrangian $L$ is defined by the couple $(x_+, x_-) \in \R^n \times \R^n$ and by their respective temporal evolutions governed by the following differential equations
\begin{eqnarray*} 
\partial_1 L(x_+(t), \mc{D}^+ x_+(t),t) - \mc{D}^+ \partial_2 L (x_+(t), \mc{D}^+ x_+(t),t) & = & 0, \nonumber \\
\partial_1 L(x_-(t), \mc{D}^- x_-(t),t) - \mc{D}^- \partial_2 L (x_-(t), \mc{D}^- x_-(t),t) & = & 0. \nonumber
\end{eqnarray*} 
\end{definition}

The asymmetric embedding of those systems is now illustrated with few operators $\mc{D}^+$ and $\mc{D}^-$.


\section{Particular cases} \label{Section:Cas_part}

We consider the same Lagrangian system $\mc{S}$ as above, with Lagrangian $L$.
First the case of the finite differences is presented. These provide an clear illustration of the topic. Then the degenerate case of the classical derivative is adressed, before finishing with the fractional operators.

\subsection{Finite differences}

For $\varepsilon > 0$ fixed, we choose $\mc{D}^\pm = d^\pm_\varepsilon$, with
\begin{eqnarray}
\Dfp f(t) & = & \dfrac{1}{\varepsilon}(f(t)-f(t-\varepsilon)), \label {Dfp} \\
\Dfm f(t) & = & \dfrac{1}{\varepsilon}(f(t+\varepsilon)-f(t)). \nonumber
\end{eqnarray} 

We verify that the operator $\Dfp$ takes into account the past of $f$ and $\Dfm$ the future. 

For $a=-\infty$ and $b=+\infty$, we have $\mc{U} = C^0(\R,\R^n)$. We verify that $C^\infty_c(\R,\R^n) \subset \mc{U}$.

We also have a relation similar to \eqref{IPP}:
\begin{equation*}
\int_\R \Dfp f(t) \, g(t) \, dt = - \int_\R f(t) \, \Dfm g(t) \, dt.
\end{equation*} 

In this case, $R_{ab}(f,g) = 0$ and $\mc{H} = \mc{V}$.

For those operators, the asymmetric dynamical representation is
\begin{eqnarray}
\partial_1 L(x_+(t), \Dfp x_+(t),t) - \Dfp \partial_2 L(x_+(t), \Dfp x_+(t),t) = 0, \label{EL_Dfp} \\
\partial_1 L(x_-(t), \Dfm x_-(t),t) - \Dfm \partial_2 L(x_-(t), \Dfm x_-(t),t) = 0. \label{EL_Dfm}
\end{eqnarray} 

\begin{remark}
Concerning the least action principle, we may note that the condition $t \mapsto \partial_2 \tilde{L} (X(t), \mc{D} X(t), t) \in \mc{U}$ is always verified in this case.
\end{remark}

This example is a good illustration of the embedding notion. When we observe (or simulate) the evolution of a classical Lagrangian system (ruled by the operator $d/dt$), we only measure its state at particular instants. The observations remain ponctual, never continuous. The velocity is calculated from those ponctual measures, with a formula similar to \eqref{Dfp}. 
Experimentally, we cannot choose $\varepsilon$ as small as we want in order to recover $d/dt$. Besides, even if we suppose the time step very small, we would approximate $d_+$ defined by \eqref{dp} and not $d/dt$. The dynamics that we observe is therefore not ruled by $d/dt$, but by $\Dfp$.

\subsection{Classical derivative}

Nevertheless, let us look at the case where the embedding does not modify the temporal evolution operator. We have $\mc{D}^+ = \mc{D}^- = d/dt$. 
We consider an interval $[a,b]$, with $-\infty \leq a < b \leq +\infty$.
In this case, $\mc{U} = C^1([a,b],\R^n)$. Once again, $C^\infty_c(\R,\R^n) \subset \mc{U}$. 
The relation \eqref{IPP} is the classical integration by parts
\begin{equation*}
\int_a^b \frac{d}{dt} f(t) \, g(t) \, dt = - \int_a^b f(t) \, \frac{d}{dt} g(t) \, dt + f(b)g(b) - f(a) g(a).
\end{equation*} 
Here we have $R_{ab}(f,g) = f(b)g(b) - f(a) g(a)$ and $\mc{H} = \{ (h_+,h_-) \in \mc{V} \, | \, h_+(a)=h_-(a)=h_+(b)=h_-(b)=0 \}$.

The asymmetric dynamical representation becomes redundant:
\begin{eqnarray*}
\partial_1 L(x_+(t), \frac{d}{dt} x_+(t),t) - \frac{d}{dt} \partial_2 L(x_+(t), \frac{d}{dt} x_+(t),t) = 0, \\
\partial_1 L(x_-(t), \frac{d}{dt} x_-(t),t) - \frac{d}{dt} \partial_2 L(x_-(t), \frac{d}{dt} x_-(t),t) = 0.
\end{eqnarray*} 

The equations are identical in both directions, are moreover similar to the reference equation \eqref{EL}. According to our approach on irreversibility, the system is in this case reversible. This point could be discussed.

\subsection{Fractional derivatives}

For a detailed presentation of the fractional calculus, we refer to \cite{Oldham,Podlubny,Samko}. For fractional derivatives, the role of the past and the future clearly appear. Regarding the different non causal Euler-Lagrange equations derived in \cite{Agrawal, Baleanu_Agrawal, Cresson, Riewe_96}, the main contribution of this paper is the obtention of a causal equation.

We consider an interval $[a,b]$, with $-\infty \leq a < b \leq +\infty$, and we consider the Riemann-Liouville fractional derivative of order $\alpha$, with $0<\alpha<1$. 
For the fractional case, precising $\mc{U}$ is not easy. We refer to \cite{Samko} for precisions. From \cite[p.159]{Samko}, we still have $C^\infty_c(\R,\R^n) \subset \mc{U}$.

Moreover, as indicated in \cite{Inizan_HFE}, we introduce a extrinsic time constant $\tau$ in order to retain the dimensional homogeneity of the equations. By the way, let us mention that the methods exposed in \cite{Inizan_HFE} naturally apply to the asymmetric embedding. Therefore it is possible to obtain fractional equations which stem from a least action principle, and which additionnally preserve causality and homogeneity. So we set $\mc{D}^+ = \tau^{\alpha-1} \Drl{a}{t}{\alpha}$ and $\mc{D}^- = - \tau^{\alpha-1} \Drl{t}{b}{\alpha}$, with
\begin{eqnarray*}
\Drl{a}{t}{\alpha} f(t) & = & \frac{1}{\Gamma(1-\alpha)} \frac{d}{du} \int_a^t (t-u)^{-\alpha} f(u) \, du, \\
\Drl{t}{b}{\alpha} f(t) & = & - \frac{1}{\Gamma(1-\alpha)} \frac{d}{du} \int_t^b (u-t)^{-\alpha} f(u) \, du.
\end{eqnarray*}

Those operators deal with the integrality of the past and the future of the function. The question of causality is hence momentous here. Relation \eqref{IPP} is now 
\begin{equation*}
\int_a^b \Drl{a}{t}{\alpha} f(t) \, g(t) \, dt = - \int_a^b f(t) \, (-\Drl{t}{b}{\alpha}) g(t) \, dt,
\end{equation*} 
which implies $R_{ab}(f,g) = 0$ and $\mc{H}=\mc{V}$. If we had chosen the Caputo fractional derivative, the only difference would have been $R_{ab}(f,g) \neq 0$.

The asymmetric dynamical representation is given here by
\begin{eqnarray}
\partial_1 L(x_+(t), \tau^{\alpha-1} \Drl{a}{t}{\alpha} x_+(t),t) - \tau^{\alpha-1} \Drl{a}{t}{\alpha} \partial_2 L(x_+(t), \tau^{\alpha-1} \Drl{a}{t}{\alpha} x_+(t),t) = 0, \label{EL_p_frac} \\
\partial_1 L(x_-(t), -\tau^{\alpha-1} \Drl{t}{b}{\alpha} x_-(t),t) + \tau^{\alpha-1} \Drl{t}{b}{\alpha} \partial_2 L(x_-(t), -\tau^{\alpha-1} \Drl{t}{b}{\alpha} x_-(t),t) = 0. \label{EL_m_frac}
\end{eqnarray}

\begin{example}
We finish with the case of the harmonic oscillator, with $a= -\infty$ and $b=+\infty$. Its Lagrangian may be written as
\begin{equation*}
L(x,v) = \frac{1}{2} \, v^2 - \omega^2 x^2.
\end{equation*} 

In this case, \eqref{EL_p_frac} and \eqref{EL_m_frac} are
\begin{eqnarray*}
\tau^{2\alpha -2} (\Drl{-\infty}{t}{\alpha} \circ \Drl{-\infty}{t}{\alpha}) x_+(t) + \omega^2 x_+(t) & = & 0,  \\
\tau^{2\alpha -2} (\Drl{t}{+\infty}{\alpha} \circ \Drl{t}{+\infty}{\alpha}) x_-(t) + \omega^2 x_-(t) & = & 0.
\end{eqnarray*} 

Given that $\Drl{-\infty}{t}{\alpha} \circ \Drl{-\infty}{t}{\alpha} = \Drl{-\infty}{t}{2\alpha}$ and $\Drl{t}{+\infty}{\alpha} \circ \Drl{t}{+\infty}{\alpha} = \Drl{t}{+\infty}{2 \alpha}$, we obtain
\begin{eqnarray*}
\tau^{2\alpha -2} \Drl{-\infty}{t}{2\alpha} x_+(t) + \omega^2 x_+(t) & = & 0,  \\
\tau^{2\alpha -2} \Drl{t}{+\infty}{2\alpha} x_-(t) + \omega^2 x_-(t) & = & 0.
\end{eqnarray*}

For $\alpha = 1$, $\Drl{-\infty}{t}{2\alpha} = \Drl{t}{+\infty}{2\alpha} = d^2 / dt^2$. We recover the classical equations
\begin{eqnarray*}
\frac{d^2}{dt^2} x_+(t) + \omega^2 x_+(t) & = & 0,  \\
\frac{d^2}{dt^2} x_-(t) + \omega^2 x_-(t) & = & 0.
\end{eqnarray*}

The dynamics is reversible, which is in agreement with the usual notion of reversibility (invariance under the transformation $t \mapsto -t$).

On the other hand, for $\alpha = 1/2$, we have $\Drl{-\infty}{t}{2\alpha} = d/dt$ and $\Drl{t}{+\infty}{2\alpha} = - d/dt$, which leads to
\begin{eqnarray*}
\frac{d}{dt} x_+(t) + \tau \omega^2 x_+(t) & = & 0,  \\
- \frac{d}{dt} x_-(t) + \tau \omega^2 x_-(t) & = & 0.
\end{eqnarray*}

Those two equations differ: the dynamics is therefore irreversible. Once again, this result agrees with the usual definition of irreversibility.

\end{example}


\section{Conclusion} \label{Section:Concl}

By distinguishing left and right derivatives, a formal approach on irreversibility has been presented. It may help clarify the role of past and future in the least action principle. In particular, trajectories and variations have been uncoupled at the dynamical level. This differentiation has led to the obtention of Euler-Lagrange equations which respect causality, even in the fractional case.

We underline once again that this point of view on irreversibility is formal and does not explain the physical origin of this phenomenon. However, it may constitute a trail for a better understanding, through the role of the temporal evolution operators. Without closing the debate on the least action pinciple, this paper precises the relation between this principle and causality. Moreover, the formalism presented here associated to that in \cite{Inizan_HFE} finally conciliates fractional Euler-Lagrange equations with fundamental physical principles. We hope that this mathematical framework, now physically satisfying, will facilitate the description of phenomena still not properly understood, particularly in the fractional case. We think for example of some Hamiltonian chaotic systems, for which fractional diffusion models have been developped in \cite{Zasl_HCFD}. In those systems, fractional dynamics might arise with the construction of a new ``macroscopic'' time \cite{Hilfer_FFD}, based on recurrence times and describing long-term behaviours.


\bibliographystyle{plain}
\bibliography{Biblio_Irreversibility}

\end{document}